\documentclass[10pt,twocolumn]{article} 

\usepackage{latex8} 

\usepackage[latin1]{inputenc}
\usepackage[T1]{fontenc}

\usepackage{times}

\usepackage{amsmath, amssymb, latexsym, amsfonts}
\usepackage{euscript, stmaryrd, bigstrut}

\usepackage{delarray}

\usepackage{epic, epsfig, pstricks, pst-node, pst-tree, pst-coil}
\usepackage{gastex}

\usepackage[english]{babel}

\usepackage{verbatim}

\usepackage{ifthen}

\usepackage{url}

\usepackage{listings}

\usepackage{galois}

\usepackage{bbold}

\usepackage{pst-plot}

\usepackage{wrapfig}

\usepackage{algorithm2e}
\usepackage{stmaryrd}
\usepackage{amsfonts}
\usepackage{amsthm}

\lstdefinelanguage{myalgo}{%
  morekeywords={if,then,else,repeat,while,for,to,forall,do,compute,call,return}
}

\title{Decomposition of Decidable First-Order Logics over Integers and Reals\thanks{Work supported by the Agence Nationale de la Recherche, grant ANR-06-SETIN-001.}}

\author{
Florent Bouchy, Alain Finkel\\
LSV, ENS Cachan, CNRS\\ 
CNRS UMR 8643, Cachan, France\\ 
\{bouchy,finkel\}@lsv.ens-cachan.fr
\and 
J{\'e}r{\^o}me Leroux\\
Laboratoire Bordelais de Recherche en Informatique\\
CNRS UMR 5800, Talence, France\\
leroux@labri.fr
}

\date{\today}

\def\dsp#1{\displaystyle{#1}}

\def\c#1{\mathcal{#1}}

\newcommand{\st}{\ |\ }

\newcommand{\moins}{\backslash}

\newcommand{\dom}[1]{\textsf{dom}(#1)}
\newcommand{\supp}[1]{\textsf{supp}(#1)}

\newcommand{\setR}{\mathbb{R}}
\newcommand{\setZ}{\mathbb{Z}}
\newcommand{\setD}{\mathbb{D}}
\newcommand{\setN}{\mathbb{N}}

\newcommand{\rondZ}{\mathfrak{Z}}
\newcommand{\rondD}{\mathfrak{D}}
\newcommand{\rondR}{\mathfrak{R}}

\newcommand{\fo}[1]{\textrm{FO}\left(#1\right)}

\newcommand{\partie}[1]{\textsf{P}(#1)}

\newcommand{\inter}[1]{\left\llbracket #1 \right\rrbracket}

\newcommand{\Lan}[1]{\textsf{Lan}(#1)}

\newcommand{\automaton}{A}

\newcommand{\scalar}[2]{\left<#1,#2\right>}

\newcommand{\udots}{\mathinner{\mskip1mu\raise1pt\vbox{\kern7pt\hbox{.}}
  \mskip2mu\raise4pt\hbox{.}\mskip2mu\raise7pt\hbox{.}\mskip1mu}} 

\newcommand{\func}{\mathcal{F}}

\newtheorem{thm}{Theorem}

\newtheorem{dfn}[thm]{Definition}

\newtheorem{prp}[thm]{Proposition}

\newtheorem{exe}[thm]{Example}

\renewcommand{\vec}{\mathbf}

\begin{document}
\maketitle

\begin{abstract}
We tackle the issue of representing infinite sets of real-valued vectors. This paper introduces an operator for combining integer and real sets. Using this operator, we decompose three well-known logics extending Presburger with reals. Our decomposition splits a logic into two parts : one integer, and one decimal (i.e. on the interval $[0,1[$). We also give a basis for an implementation of our representation.
\end{abstract}

\section{Introduction}

Verification (and model-checking in particular) of infinite systems like timed automata \cite{AD94} (and hybrid systems) and counter systems \cite{DBLP:conf/tacas/BardinFL04} need good symbolic representation classes ; by \emph{good}, we mean having closure properties (by first-order logic operators) and decidability results (for testing inclusion and emptiness). Presburger arithmetic \cite{Presburger,GinSpa66} enjoys such good properties, and some efficient implementations (using finite automata) have been intensively used for the analysis of counter systems \cite{BFLP03,DFvB05,BJW01,BJW03}.\\

Despite the fact that the complete arithmetic on reals is decidable \cite{Tarski48}, only some restricted classes of the first-order additive logic of reals (DBM, CPDBM, finite unions of convex polyhedra) have been used for the analysis of timed automata. This is mainly due to the fact that the algorithmic complexity of DBM is polynomial, which is the basis of efficient verification algorithms for timed automata in \textsc{UppAal} \cite{UppAal95,UppAal97}.\\

However, we would like to be able to use both integers and reals, for at least two reasons. First, we want to analyse timed counter systems \cite{AAB00,TReX,BBR97} in which the reachability sets contain vectors with both integers and reals. Second, we want to be able to use integers as parameters for a concise representation of pure reals : for instance, reals are used for the values of clocks and integers for expressing the parameters in CPDBM.\\

Fortunately, the first-order additive logic over integers and reals is decidable. Nevertheless, the algorithmic of sets combining integers and reals does not seem simple, even when it is based on finite automata like Real Vector Automata \cite{BBR97,BRW98} or weak RVA \cite{LIRA}, or based on quantifier elimination \cite{Weispfenning99}.\\

For that matter, the algorithmic of Presburger (using finite automata) and variations of DBM are quite efficient. Hence, our idea is to reduce the algorithmic difficulty of the first-order additive logic of integers and reals (and of some subclasses and decidable extensions) by decomposing a complex set of integers and reals into a finite union of sums of integer sets and decimal sets. By decimal, we mean numbers in the dense inteval $[0,1[$ ; then, we define a new class of sets as follows. Given $n$ sets of integers $(Z_i)_{0\leq i\leq n}$ and $n$ sets of decimals $(D_i)_{0\leq i\leq n}$, we introduce the operator \emph{ finite union of sums}, which builds the finite unions of the sums $Z_i + D_i$. This class is shown stable under boolean operations, cartesian product, quantification and reordering if both of the two initial classes are also stable.\\

One of our aims is then to re-use, in combining the best representations of these two initial sets $(Z_i)_{0\leq i\leq n}$ and $(D_i)_{0\leq i\leq n}$, the best libraries dealing with them to efficiently handle finite unions of $(Z_i+D_i)_{0\leq i\leq n}$ (for instance : \textsc{PresTAF} \cite{BLP06} for the integers and \textsc{PPL} \cite{PPL} for the reals).\\

We show that three of the main classes of mixed integer and real sets are in fact finite unions of sums of well-known classes. We prove that finite unions of sums of Presburger set of integers, and sets definable in the first-order additive logic of decimals are exactly the sets definable in the first-order logic of integers and reals. The finite unions of CPDBM are expressible as the finite unions of sums of Presburger-definable sets and DBM-definable decimal sets. Moreover, when we go beyond Presburger by considering RVA, we show that the class of sets representable by RVA in basis $b$ is the finite unions of sums of Presburger extended with a predicate $V_b$ (which gives integer powers in base $b$) and the additive logic of decimals extended with a predicate $W_b$ (which, similarily to $V_b$, gives negative powers in base $b$).\\

\section{Representations mixing integers and reals}\label{sec:operateur}

In this section, we motivate our work with a small example of timed automaton. We show that extracting integers from reals can yield more concise formula than pure reals. Then we introduce an operator combining integer and real sets of vectors.\\

\subsection{Timed Automata and DBM}

In order to study real-life systems involving behaviours that depend on time elapsing, timed automata are probably the most used and well-known model for such systems. As described in \cite{AD94}, the basic idea of timed automata is to add real-valued variables (called clocks) to finite automata. These clocks model temporal behaviours of the system, flowing at a universal constant rate~; each clock can be compared to an integer constant, and possibly reset to $0$. The only other guard allowed is called a diagonal constraint, consisting in comparing the difference of two clocks to an integer constant. As the clocks' values are unbounded, the state-space generated by a timed automaton is infinite~; therefore, regions are used to model a finite abstraction of the system's behaviour. Practically intractable because of its size, the region graph is then implemented as zones in most verification tools \cite{UppAal95,UppAal97,Kronos,CMC} modelling such real-time systems.\\

Technically, zones are represented by Difference Bound Matrices (DBM) \cite{BM83_DBM,Dill89} in these tools. A DBM is a square matrix representing the constraints between $n$ clocks defining a zone. Here, we see a DBM as a tuple $(\vec{c},\vec{\prec})$, where $\vec{c} = (c_{i,j})_{0\leq i,j\leq n}$, $\vec{\prec}=(\prec_{i,j})_{0\leq i,j\leq n}$, $c_{i,j}\in \setZ\cup\{+\infty\}$, and $\prec_{i,j}\in \{\leq,<\}$. Each element of this tuple is an element of the square matrix, defining a DBM set as follows :
$$R_{\vec{c},\vec{\prec}}=\{\vec{r}\in\setR^n \st \bigwedge_{0\leq i,j\leq n}r_i-r_j\prec_{i,j}c_{i,j}\}$$
In order to deal with constraints involving only one clock, the fictive clock $r_0$ is always set to the value $0$. An element $(c_{i,j},\prec_{i,j})$ means that $r_i - r_j \prec_{i,j} c_{i,j}$, where $r_i,r_j$ are clocks. Thus, each element of a DBM represents a diagonal constraint (i.e. a bounded difference). Finally, terms that do not represent any actual constraint are symbolized by $c_{i,j}=+\infty$.\\

\subsection{About extensions of DBM}

On the following example taken from \cite{BBFL-tacas-2003}, the timed automaton features 2 clocks $x$ and $y$, and a unique location. The automaton's behaviour is very simple~: $y$ is reset to $0$ as soon as it reaches $1$, while $x$ flows continually. In the initial state, the clocks are both set to $0$. Moreover, an invariant in the location ensures that $y$ never exceeds $1$.

  
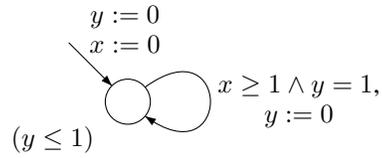
\begin{center}
    \begin{picture}(54,20)(-17,-6)
      \gasset{Nw=6,Nh=6,Nmr=4}
      \node(A)(0,0){}
      \node[Nframe=n](B)(-10,10){}
      \node[Nframe=n](C)(-10,-5){$(y \le 1)$}
      \drawedge(B,A){\shortstack{$y:=0$\\$x:=0$}}
      \drawloop[loopangle=0](A){\shortstack{$x \ge 1 \wedge y = 1$,\\$y:=0$}}
    \end{picture}
  \end{center}
\medskip
  
The clock diagram associated to the automaton explicitely shows this behaviour~:
\smallskip

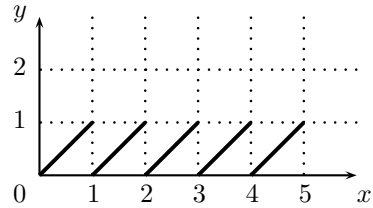
  \begin{center}
  	\psset{unit=4pt}
    \begin{pspicture}(33,17)
      \psline{->}(0,0)(30,0)
      \psline{->}(0,0)(0,15)
      \psset{linewidth=1.5pt}
      \psline(0,0)(5,5)
      \psline(5,0)(10,5)
      \psline(10,0)(15,5)
      \psline(15,0)(20,5)
      \psline(20,0)(25,5)
      \psset{linestyle=dotted,linewidth=1pt}
      \psline(5,0)(5,15)
      \psline(10,0)(10,15)
      \psline(15,0)(15,15)
      \psline(20,0)(20,15)
      \psline(25,0)(25,15)
      \psline(0,5)(30,5)
      \psline(0,10)(30,10)
      \put(-2.5,-2.5){$0$}
      \put(4.5,-2.5){$1$}
      \put(9.5,-2.5){$2$}
      \put(14.5,-2.5){$3$}
      \put(19.5,-2.5){$4$}
      \put(24.5,-2.5){$5$}
      \put(30,-2.5){$x$}
      \put(-2.5,4.5){$1$}
      \put(-2.5,9.5){$2$}
      \put(-2.5,15){$y$}
    \end{pspicture}
  \end{center}

\medskip

A classical forward analysis \cite{BLR05} is considered here, by computing the reachable states (i.e. $location \times clock\ values$) from the initial one (where $x=y=0$). Then, we build the corresponding zones, each zone being represented by a DBM~; here, we have an infinite yet countable set of DBM as follows. Note that in this example $\prec$ is always $\leq$~; therefore, we will omit it in the matrices.
\[
\left\lbrace
\begin{array}{rl}
\vspace{-.15cm}
 & \begin{array}{c}
	\begin{array}{r}\hspace{-.3cm}_0\end{array}
	\begin{array}{r}\hspace{.43cm}_x\end{array}
	\begin{array}{r}\hspace{.25cm}_y\end{array}
\end{array} \vspace{.15cm} \\
\hspace{-.3cm} \begin{array}{r}
	\begin{array}{r} _0\\ _x\\ \vspace{.15cm}_y\end{array}
\end{array} & \hspace{-.5cm} \begin{pmatrix} 0 & -i & 0 \\ i+1 & 0 & i \\ 1 & -i & 0 \end{pmatrix}\\ 
\end{array} \hspace{-.15cm}
\right\rbrace_{i\geq 0}
\]

In order to make the state-space computable, abstraction techniques are used to get a finite number of zones. The abstraction being used in most model-checkers is based on maximum constants~: a clock $c$'s valuation is considered equal to $\infty$ as soon as it exceeds the maximal constant to which $c$ is ever compared. On the example, if a guarded transition $x\geq 10^6$ leads to another state, then the clock diagram becomes as follows~:

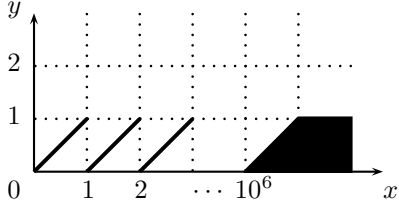
  \begin{center}
    \psset{unit=4pt}
    \begin{pspicture}(36,17)
      \psline{->}(0,0)(33,0)
      \psline{->}(0,0)(0,15)
	\pspolygon[linewidth=1pt,fillstyle=solid,fillcolor=black](20.1,0)(30,0)(30,5.1)(25,5.1)
      \psset{linewidth=1.5pt}
      \psline(0,0)(5,5)
      \psline(5,0)(10,5)
      \psline(10,0)(15,5)
      \psline(20,0)(25,5)
      \psset{linestyle=dotted,linewidth=1pt}
      \psline(5,0)(5,15)
      \psline(10,0)(10,15)
      \psline(15,0)(15,15)
      \psline(20,0)(20,15)
      \psline(25,0)(25,15)
      \psline(0,5)(30,5)
      \psline(0,10)(30,10)
      \put(-2.5,-2.5){$0$}
      \put(4.5,-2.5){$1$}
      \put(9.5,-2.5){$2$}
      \put(15,-2.5){$\cdots$}
      \put(19,-2.5){$10^6$}
      \put(33,-2.5){$x$}
      \put(-2.5,4.5){$1$}
      \put(-2.5,9.5){$2$}
      \put(-2.5,15){$y$}
    \end{pspicture}
  \end{center}

\smallskip

More formally, this abstraction yields the following set of DBM~:\\
\[
\hspace{-.1cm}
\left\lbrace
\hspace{-.25cm}
\begin{array}{ll}
\left\lbrace
\begin{array}{rl}
\vspace{-.15cm}
 & \begin{array}{c}
	\begin{array}{r}\hspace{-.45cm}_0\end{array}
	\begin{array}{r}\hspace{.4cm}_x\end{array}
	\begin{array}{r}\hspace{.25cm}_y\end{array}
\end{array} \vspace{.15cm} \\
\hspace{-.5cm} \begin{array}{r}
	\begin{array}{r} _0\\ _x\\ \vspace{.15cm}_y\end{array}
\end{array} & \hspace{-.7cm} \begin{pmatrix} 0 & -i & 0 \\ i+1 & 0 & i \\ 1 & -i & 0 \end{pmatrix}\\
\end{array} \hspace{-.2cm}
\right\rbrace_{\hspace{-.13cm} 0 \leq i \leq 10^6},
&
\hspace{-.2cm}
\begin{array}{rl}
\vspace{-.15cm}
 & \begin{array}{c}
	\begin{array}{r}\hspace{-.45cm}_0\end{array}
	\begin{array}{r}\hspace{.4cm}_x\end{array}
	\begin{array}{r}\hspace{.25cm}_y\end{array}
\end{array} \vspace{.15cm} \\
\hspace{-.5cm} \begin{array}{r}
	\begin{array}{r} _0\\ _x\\ \vspace{.15cm}_y\end{array}
\end{array} & \hspace{-.7cm} \begin{pmatrix} 0 & \infty & 0 \\ \infty & 0 & \infty \\ 1 & -10^6 & 0 \end{pmatrix}\\ 
\end{array} \hspace{-.15cm}
\end{array} \hspace{-.4cm}
\right\rbrace
\]

This set of DBM is finite, but remains huge~: $10^6+2$ matrices need to be computed and memorized, which seems exaggerated, a fortiori for such a simple example. In \cite{BBFL-tacas-2003}, a more elaborate abstraction is proposed~: the clocks' maximal constants are no more global to the system, but location-dependent. Another abstraction technique is proposed in \cite{BBLP-STTT05}, distinguishing between upper and lower bounds within maximal constants. To the best of our knowledge, these are the only zone-based abstraction techniques~; in each of them, the number of DBM still heavily depends on maximal constants.\\

Writing here such an infinite or huge number of DBM would have been impossible~; therefore, we naturally used a parametric representation of these DBM. Actually, this idea is also used by Constrained Parametric DBM (CPDBM) \cite{AAB00}, which is the data structure implemented in the \textsc{TReX} \cite{TReX} model-checker. CPDBM are indeed a more expressive version of DBM, extended in two steps. First, we consider PDBM, in which $c_{i,j}$ constants become $t_{i,j}$ arithmetical terms (the parameters). Such arithmetical terms $t$ are given by the grammar $t ::= 0 \mid 1 \mid x \mid t-t \mid t+t \mid t*t$, where $x$ belongs to a set $\mathcal{X}$ of real variables. Second, a PDBM becomes a CPDBM as terms are constrained by quantifier-free first-order formulas $\phi$. Such formulas are defined by $\phi ::= t\leq t \mid \neg \phi \mid \phi \vee \phi \mid Is\_int(t)$ (where the predicate $Is\_int(t)$ is true iff $t$ is an integer). Each of the two sets of matrices hereinabove is in fact a single CPDBM.\\

\def\lishape{
\psset{unit=1pt} 
\begin{pspicture}(10,10)
 \psline(0,0)(8,8)
 \psset{linestyle=dotted} 
 \psline(8,8)(8,0)
 \psline(8,0)(0,0)
 \psline(0,0)(0,8)
 \psline(0,8)(8,8)
\end{pspicture}}

\def\trshape{
\psset{unit=1pt} 
\begin{pspicture}(10,10) 
 \pspolygon[fillstyle=solid,fillcolor=black](0,0)(8,0)(8,8)
 \psset{linestyle=dotted} 
 \psline(0,0)(0,8)
 \psline(0,8)(8,8)
\end{pspicture}}

\def\sqshape{
\psset{unit=1pt} 
\begin{pspicture}(10,10)
 \pspolygon[fillstyle=solid,fillcolor=black](0,0)(8,0)(8,8)(0,8)
\end{pspicture}}

Consider now another way to represent the set of reachable clock values. On the second diagram showing the abstraction, we can see an obvious regular pattern along $x$, defined by three shapes~: \lishape , \trshape , and \sqshape . We define each shape as follows~: \lishape $= \lbrace (x,y)\in[0,1]^2 \st x=y \rbrace$, \trshape $= \lbrace (x,y)\in[0,1]^2 \st x\geq y \rbrace$, and \sqshape $= \lbrace (x,y)\in[0,1]^2\rbrace$. If we want to represent the same set as the previous abstracted zones, but without DBM, we can express the periodicity of each pattern with integers. To formalize it, taking the union of the following three sums suffices~:
\begin{equation*}
\begin{aligned}
& \Bigl(\{0,\ldots ,10^6-1\} \times \{0\} + \lishape \Bigr) \\
\bigcup & \Bigl(\{10^6\} \times \{0\} + \trshape \Bigr) \\
\bigcup & \Bigl(\{10^6+1,\ldots ,\infty \} \times \{0\} + \sqshape \Bigr) \\
\end{aligned}
\end{equation*}

This latter symbolic representation of such a reachability set is much smaller than DBM. Indeed, representing zones with DBM implies memorizing a possibly huge number of matrices, depending on the maximal constant for the clocks (one million, in this example). However, by introducing integers to express periodicity, we can reduce the representation to three small combinations of intervals. Moreover, we can even get rid of the abstraction, so as to get an exact representation for the same cost. CPDBM also have these advantages, but are undecidable because of the multiplication. Hence, let us specify a little more what is our representation~: we take finite unions of reals, real numbers being decomposed as sums of integers and smaller reals (called decimals). These integers and reals can be defined using quantification, addition, and boolean operators.\\

Actually, our approach comes down to representing sets of real numbers by extracting their integer components~; the interesting point is that adding integers to real sets can simplify their representation and ease their handling. One might think that adding integers to such a first-order real logic would make it undecidable, but section \ref{sec:presburger} proves the opposite. Before that, we need to formalize our representation.\\

\subsection{Composing integers and reals}
\label{sec:composing}

\paragraph{Notations.}
The set $\left[0,1\right[$ is denoted by $\setD$ in the sequel. We also call a \emph{decimal} (number) any $d \in \setD$, and a \emph{decimal set} any $D \subseteq \setD$. We write $\vec{x}$ to denote a \emph{vector} $(x_1,\dots,x_n)$. Sometimes, in order to be concise, we use $\fo{\dots}$ to denote the sets represented by this first-order logic. However, it does not make our statements incorrect, because we mostly discuss the expressive power of such logics.\\

Let $\rondZ \subseteq \partie{\setZ^n}$ and $\rondD \subseteq \partie{\setD^n}$ ; we will assume in this paper that we are using $n-$dimensional vectors, with $n \in \setN$. We denote by\footnote{The symbol $\uplus$ is sometimes used for the disjoint union, but we do not use such unions in this paper.} $\rondZ\uplus\rondD$ the class of real vectors $R \subseteq \setR^n$ s.t. $\dsp R=\bigcup_{i=1}^p(Z_i+D_i)$, with $(Z_i,D_i) \in \rondZ\times\rondD$ and $p\geq 1$. \\

Here are some examples of simple sets that might be often used, written as finite unions of sums of integers and decimals~:

\begin{exe}\label{ex:set}
  The empty set $\emptyset$ is written $\emptyset + \emptyset $. The set $\setR^n $ is written $\setZ^n + \setD^n$. The set $\setZ^n$ is written $\setZ^n+\{\vec{0}\}$.
\end{exe}

\begin{exe}\label{ex:eq}
The set $R_==\{\vec{r}\in\setR^2 \st r_1= r_2\}$ is written $\{\vec{z}\in \setZ^2\st z_1=z_2\} + \{\vec{d}\in \setD^2\st d_1=d_2\}$
\end{exe}

\begin{exe}\label{ex:leq}
  The set $R_\leq=\{\vec{r}\in\setR^2 \st r_1\leq r_2\}$ is written : 
\begin{align*}
& \{\vec{z}\in \setZ^2\st z_1\leq z_2\} + \{\vec{d}\in \setD^2\st d_1\leq d_2\}\\
 \bigcup &\{\vec{z}\in \setZ^2\st z_1<z_2\} + \{\vec{d}\in \setD^2\st d_1 > d_2\}
\end{align*}
\end{exe}

\begin{exe}\label{ex:sum}
  The set $R_+=\{\vec{r}\in\setR^3 \st r_1+r_2=r_3\}$ is written $\bigcup_{c\in \{0,1\} }\{\vec{z}\in \setZ^3\st z_1+z_2+c=z_3\} + \{\vec{d}\in \setD^3\st d_1+d_2=d_3+c\}$, where $c$ denotes a carry.
\end{exe}

The limits of our representation can be seen with the following counter-example\label{ex:infini}. Consider the set $\dsp R=\bigcup_{j=1}^{\infty} \Bigr(\lbrace j\rbrace + \left\lbrace \frac{1}{j+1}\right\rbrace \Bigl) $~; note that we use $j+1$ (and not simply $j$) to avoid the case where the decimal part is $\frac{1}{j}=1$ for $j=1$ (because it would not be a decimal, i.e. in $[0,1[$). Our representation can not deal with such a set~; indeed, despite the fact that it is a union of sums of integers and decimals, we can see that the union is inherently infinite. We insist on the finiteness of the union in our representation, mainly for implementability reasons~; this will be discussed in section \ref{sec:idf}.\\

Now, let us consider the stability of our representation. We prove\footnote{Here we have to take unions, depending on the number of dimensions, for a technical purpose~: the projection of a component in the vector.} that if  $\rondZ\subseteq\bigcup_{n\in\setN}\partie{\setZ^n}$ and $\rondD\subseteq \bigcup_{n\in\setN}\partie{\setD^n}$ are stable by the classical first order operations then the class $\rondZ\uplus\rondD=\bigcup_{n\in\setN}\rondZ_n\uplus\rondD_n$ where $\rondZ_n=\rondZ\cap\partie{\setZ^n}$ and $\rondD_n=\rondD\cap\partie{\setD^n}$ is also stable by these operations. The operations we consider are~: boolean combinations (union, intersection, difference), cartesian product, quantification, and reordering. We use the following definitions for these last two operations. First, quantification is done by projecting away variables from the considered vector~: $\forall R \subseteq \setR^n,\ \exists_i R = \{(r_1,\dots , r_{i-1}, r_{i+1}, \dots , r_n) \mid \exists r_i\ (r_1,\dots, r_{i-1}, r_i, r_{i+1}, \dots ,r_n)\in R\}$. Second, a reordering is a mere permutation function $\pi$ of the variables order in a vector~: $\forall R \subseteq \setR^n,\ \pi R = \{(r_{\pi (1)},\dots,r_{\pi (n)}) \mid (r_1,\dots,r_n) \in R\}$. Then, we introduce a generic definition for stability~:

\begin{dfn}
  A class $\rondR\subseteq \bigcup_{n\in\setN}\partie{\setR^n}$ is \emph{stable} if it is closed under boolean operations, cartesian product, quantification, and reordering.
\end{dfn}

Notice that taking the union of two such sets is trivial, as they are already unions of integer and decimal parts. Then, observe that $(Z_1+D_1)\cap (Z_2+D_2)=(Z_1\cap Z_2)+(D_1\cap D_2)$ for any $Z_1,Z_2\subseteq \setZ^n$ and for any $D_1,D_2\subseteq\setD^n$~; thus, the stability by union of $\rondZ_n\uplus\rondD_n$ provides the stability by intersection. From the equality $(Z_1+D_1)\moins (Z_2+D_2)=((Z_1\moins Z_2)+D_1)\cup (Z_1+(D_1\moins D_2))$ we get the stability by difference. The stability by cartesian product is provided by $(Z_1+D_1)\times (Z_2+D_2)=(Z_1\times Z_2)+(D_1\times D_2)$. The stability by projection comes from $\exists_i R=(\exists_i Z)+(\exists_i D)$, where $R=Z+D$. Finally, the stability by reordering is obtained thanks to $\pi(Z+D)=(\pi Z)+(\pi D)$.
We have proved the following proposition, which is later used in the proofs of theorem \ref{thm:pres} and proposition \ref{prp:rva}~: 
\begin{prp}[Stability]\label{prp:stable}
The class $\rondZ\uplus\rondD$ is stable if $\rondZ$ and $\rondD$ are stable.
\end{prp}

\section{First-order additive logic over integers and reals}\label{sec:presburger}

Using at the same time integers and reals in the whole arithmetic is known to be undecidable. However, when multiplication is left apart, the first-order additive logic is decidable ; its decidability has been suggested by B{\"u}chi, then proved by \cite{BRW98} with automata and by \cite{Weispfenning99} using quantifier elimination. Actually, it can be seen as the Presburger logic \cite{Presburger} extended to the reals. This first-order logic $\fo{\setR,\setZ,+,\leq}$ can encode complex linear constraints combining both integral and real variables. In this section we prove that sets definable in this logic can be decomposed into finite unions of $Z+R$ where $Z$ is definable in $\fo{\setZ,+,\leq}$ and $R$ is definable in $\fo{\setD,+,\leq}$. This result proves that complex linear constraints combining integral and real variables can be decomposed into linear constraints over integers, and linear constraints over reals. More precisely, we prove the following decomposition~:

\medskip

\begin{thm}\label{thm:pres}
  $\fo{\setR,\setZ,+,\leq}=\fo{\setZ,+,\leq}\uplus\fo{\setD,+,\leq}$.
\end{thm}
\begin{proof}
  First of all, observe that any set definable in the logic $\fo{\setZ,+,\leq}\uplus\fo{\setD,+,\leq}$ is also definable in $\fo{\setR,\setZ,+,\leq}$. Conversely, the sets $\setR$ and $\setZ$, the function $+:\setR\times\setR\rightarrow\setR$ and the predicate $\leq$ are definable in $\fo{\setZ,+,\leq}\uplus\fo{\setD,+,\leq}$ from examples \ref{ex:set}, \ref{ex:eq}, \ref{ex:leq}, \ref{ex:sum}. Thus, stability by first order operations provides the inclusion $\fo{\setR,\setZ,+,\leq}\subseteq \fo{\setZ,+,\leq}\uplus\fo{\setD,+,\leq}$. We deduce the equality.
\end{proof}


Now, let us recall that sets definable in the Presburger logic $\fo{\setZ,+,\leq}$ can be characterized thanks to \emph{linear sets} \cite{GinSpa66}. In fact, a set $Z\subseteq\setZ^n$ is definable in this logic if and only if it is equal to a finite union of linear sets $\vec{b}+P^*$ where $\vec{b}\in\setZ^n$, $P$ is a finite subset of $\setZ^n$, and $P^*$ denotes the set of finite sums $\sum_{i=1}^k p_i$ with $p_1,\ldots,p_k\in P$ and $k\in\setN$. This geometrical characterization can be extended to the class of sets definable in $\fo{\setZ,+,\leq}\uplus\fo{\setD,+,\leq}$ by introducing the class of \emph{polyhedral convex sets}. A set $C\subseteq\setR^n$ is said \emph{polyhedral convex} if $C$ is defined by a finite conjunction of formulas $\scalar{\alpha}{\vec{x}}\prec c$ where $\alpha\in\setZ^n$, $\prec\in\{\leq,<\}$ and $c\in\setZ$. Recall that a \emph{Fourier-Motzkin quantification elimination} proves that a set $C\subseteq\setR^n$ is definable in $\fo{\setR,+,\leq}$ if and only if it is equal to a finite union of polyhedral convex sets. In \cite{FL08}, the authors have proved the following geometrical characterization : \\
\textit{A set $R\subseteq \setR^n$ is definable in $\fo{\setR,\setZ,+,\leq}$ if and only if it is equal to a finite union of sets of the form $C+P^*$ where $C\subseteq\setR^n$ is a polyhedral convex set and $P$ is a finite subset of $\setZ^n$.}


\subsection{Decomposing DBM-based representations}\label{sec:cpdbm}

In this section, we characterize an extension of DBM. We denote by $\bigcup\textrm{DBM}_\setD$ the finite unions of DBM sets which are included in $\setD^n$. Notice that $\bigcup\textrm{DBM}_\setD$ is stable by first order operations, thanks to a Fourier-Motzkin quantifier elimination.

\medskip

A CP-DBM$_L$ is a DBM where the vector $\vec{c}$ is no longer a constant, but a vector of parameters constrained by a formula $\phi (\vec{c})$ defined in a logic $L$. More precisely, a CP-DBM$_L$ is a tuple $(\phi,\vec{\prec})$ representing a set $R_{\phi,\vec{\prec}}$ s.t. :
$$R_{\phi,\vec{\prec}}= \bigcup_{\vec{c} \models \phi} R_{\vec{c},\vec{\prec}} $$ As introduced in \cite{AAB00}, CPDBM correspond to CP-DBM$_L$ where $L$ is the first-order arithmetic without quantifiers ; in particular, multiplication is allowed in this formalism. In this section, we study another variation of DBM : CP-DBM$_+$, which is CP-DBM$_L$ where $L$ is the decidable Presburger logic $\fo{\setZ,+,\leq}$. That is, CP-DBM$_+$ are CPDBM with quantifiers but without multiplication. We denote by $\bigcup\textrm{CP-DBM}_+$ the finite unions of $R_{\vec{\phi},\vec{\prec}}$, i.e. finite unions of CP-DBM$_+$ sets.\\

We show that finite unions of CP-DBM$_+$ sets are in fact a combination of Presburger-definable sets and finite unions of DBM decimal sets : 
\begin{prp}\label{prp:cpdbm}
We have 
$\bigcup\textrm{CP-DBM}_+=\fo{\setZ,+,\leq}\uplus \bigcup \textrm{DBM}_\setD$.
\end{prp}
\begin{proof}
Let us first prove the inclusion $\supseteq$. Let us consider a DBM $(\vec{c},\vec{\prec})$ denoting a set $D\subseteq\setD^n$ and a Presburger formula $\psi(\vec{x})$ denoting a set $Z\subseteq\setZ^n$ and let us prove that $Z+D$ is a $\bigcup\textrm{CP-DBM}_+$ set. Observe that $\vec{r}\in Z+D$ if and only if there exists $\vec{z}\in Z$ such that $\vec{r}-\vec{z}\in D$. The condition $\vec{r}-\vec{z}\in D$ is equivalent to $\bigwedge_{0\leq i,j\leq n}r_i-r_j\prec_{i,j} c_{i,j}+z_i-z_j$. Let us consider the Presburger formula $\psi(\vec{p}):=\exists \vec{z}\in\setZ^n\ p_{i,j}=c_{i,j}+z_i-z_j$ and observe that $R_{\psi,\vec{\prec}}=Z+D$. We have proved the inclusion $\supseteq$.

\smallskip

For the converse inclusion, let us consider a $\textrm{CP-DBM}_+$ set $R_{\vec{\phi},\vec{\prec}}$. Let $Z_{\vec{d}}=\setZ^n\cap (R_{\vec{\phi},\vec{\prec}}-\vec{d})$ indexed by $\vec{d}\in\setD^n$. Observe that $Z_{\vec{d}}$ is actually the following set of vectors :
  $$Z_{\vec{d}} = \bigcup_{\vec{c} \models \phi} \left\lbrace \vec{z} \in \setZ^n \vert \bigwedge_{0\leq i,j\leq n} z_i - z_j \prec_{i,j} c_{i,j} + (d_j - d_i) \right\rbrace$$
  Since $d_j-d_i\in\left]-1,1\right[$ and $z_i-z_j,c_{i,j}\in\setZ$ we deduce that $z_i-z_j\prec_{i,j}c_{i,j}+(d_j-d_i)$ is equivalent to $z_i-z_j\leq c_{i,j}$ if $d_i-d_j \prec_{i,j}0$ and it is equivalent to $z_i-z_j\leq c_{i,j}-1$ otherwise. Given a matrix $\vec{m}=(m_{i,j})_{0\leq i,j\leq n}$ such that $m_{i,j}\in\{0,1\}$ for any $0\leq i,j\leq n$, we denote by $I_{\vec{m}}$ and $D_{\vec{m}}$ the following sets:
 \begin{align*}
 I_{\vec{m}}&=\{ \vec{z}\in\setZ^n \mid \exists \vec{c}\ \phi(\vec{c})\wedge\hspace{-0.2cm}\bigwedge_{0\leq i,j\leq n}z_i-z_j\leq c_{i,j}-m_{i,j} \}\\
 D_{\vec{m}}&=\{\vec{d}\in\setD^n \mid \bigwedge_{0\leq i,j\leq n} (d_i-d_j\prec_{i,j}0 \Longleftrightarrow m_{i,j}=0)\}
\end{align*}
 Note that $D_{\vec{m}}$ is a DBM set and $Z_{\vec{d}}=I_{\vec{m}}$ for any $\vec{d}\in D_{\vec{m}}$. From $\bigcup_{\vec{m}}D_{\vec{m}}=\setD^n$ we deduce that $R_{\vec{\phi},\vec{\prec}}=\bigcup_{\vec{d}\in\setD^n}Z_{\vec{d}}+\{\vec{d}\}=\bigcup_{\vec{m}}I_{\vec{m}}+D_{\vec{m}}$. We have proved that $R_{\vec{\phi},\vec{\prec}}$ is definable in $\fo{\setZ,+,\leq} \uplus \bigcup\textrm{DBM}_\setD$.
\end{proof}



\section{Beyond Presburger}\label{sec:rva}

We have just shown our decomposition to be working on $\fo{\setR,\setZ,+,\leq}$ and below. Now, we prove that it can also be used on more expressive logics. We take the example of Real Vector Automata (RVA) \cite{BRW98}, which is, to the best of our knowledge, the most expressive decidable implemented representation for sets of real and integer vectors. RVA are used in the tool \textsc{LASH} \cite{BJW01,BJW03}. In this section, the class of sets representable by RVA is proved decomposable into our formalism.
\medskip

Let $b\geq 2$ be an integer called the \emph{basis of decomposition}. We denote by $\Sigma_b=\{0,\ldots,b-1\}$ the finite set of \emph{digits} and by $S_b=\{0,b-1\}$ the set of \emph{sign digits}. An infinite word $\sigma=\vec{s}\vec{a}_1\ldots\vec{a}_k\star\vec{a}_{k+1}\vec{a}_{k+2}\ldots$ over the alphabet $\Sigma_b^n\cup\{\star\}$ is said \emph{$b$-correct} if $\vec{s}\in S_b^n$ and $\vec{a}_i\in\Sigma_b^n$ for any $i\geq 1$. In this case, $\sigma$ is called a \emph{most significant digit first decomposition} of the following real vector $\rho_b(\sigma)\in \setR^n$:
  $$\rho_b(\sigma)=b^k\left(\frac{\vec{s}}{1-b}+\sum_{i\geq 1}b^{-i}\vec{a}_i\right)$$
A \emph{Real Vector Automaton (RVA)} in basis $b$ is a B{\"u}chi automaton $\automaton$ over the alphabet $\Sigma_b^n\cup\{\star\}$ such that the language $\Lan{\automaton}$ recognized by $\automaton$ contains only $b$-correct words. The set $\inter{\automaton}$ represented by $\automaton$ is defined by $\inter{\automaton}=\{\rho_b(\sigma) \st \sigma\in\Lan{\automaton}\}$. A set $R\subseteq\setR^n$ is said \emph{$b$-recognizable} if there exists a RVA $\automaton$ in basis $b$ such that $R=\inter{\automaton}$. 

\medskip

According to \cite{BRW98}, the class of $b$-recognizable sets can be logically characterized by $\fo{\setR,\setZ,+,\leq,X_b}$ where $X_b$ is an additional predicate. The predicate $X_b$ over $\setR^3$ is such that $X_b(x,u,a)$ is true if and only if there exists a most significant digit first decomposition $\sigma=s a_1\ldots a_k\star a_{k+1}\ldots$ of $x$ and an integer $i\in\setN$ such that $a_i=a$ and $u=b^{k-i}$. 
\begin{thm}\cite{BRW98}
  A set $R\subseteq\setR^n$ is $b$-recognizable if and only if it is definable in $\fo{\setR,\setZ,+,\leq,X_b}$. 
\end{thm}

\medskip

In order to provide a decompostion of $\fo{\setR,\setZ,+,\leq,X_b}$, the predicate $X_b$ is proved expressible by two valuation functions $V_b$ and $W_b$ where~:
\begin{itemize}
\item $V_b:\setZ\moins\{0\}\rightarrow\setZ$ is the \emph{integer valuation function} introduced in \cite{Bruyere94} and defined by $V_b(z)=b^j$, where $j\in\setZ$ is the greatest integer such that $b^{-j}z\in\setZ$.
\item $W_b:\setD\moins\{0\}\rightarrow \setD$ is the \emph{decimal valuation function} defined by $W_b(d)=b^j$, where $j\in\setZ$ is the least integer such that $b^{-j}d\not\in\setD$.
\end{itemize}

By expressing $X_b$ in $\fo{\setR,\setZ,+,\leq,V_b,W_b}$ and $V_b,W_b$ in $\fo{\setR,\setZ,+,\leq,X_b}$ we deduce that $\fo{\setR,\setZ,+,\leq,X_b}=\fo{\setR,\setZ,+,\leq,V_b,W_b}$. Finally, from proposition \ref{prp:stable} and theorem \ref{thm:pres}, we get the following proposition.
\begin{prp}\label{prp:rva}
  $\fo{\setR,\setZ,+,\leq, X_b}=\fo{\setZ,+,\leq, V_b}\uplus \fo{\setD,+,\leq, W_b}$.
\end{prp}

Moreover, it is clear that the logic $\fo{\setZ,+,\leq,V_b}\uplus\fo{\setD,+,\leq,W_b}$ extends $\fo{\setZ,+,\leq}\uplus\fo{\setD,+,\leq}$. However, even if the function $W_b$ is crucial to logically characterize the class of $b$-recognizable sets, this predicate is not used in practice. In fact, in order to get efficient algorithms for manipulating B{\"u}chi automata (more precisely, minimization and determinization), we only consider sets $R\subseteq\setR^n$ that can be represented by a \emph{weak RVA} \cite{BJW01}. Recall that a B{\"u}chi automaton $\automaton$ is said \emph{weak} if any strongly connected component $S$ satisfies $S\subseteq F$ or $S\cap F=\emptyset$, where $F$ is the set of accepting states. Unfortunately, the class of sets $R\subseteq \setR^n$ representable by a weak RVA is not logically characterized since this class is not stable by first order operations (because of projection). In practice, since any set $R\subseteq\setR^n$ definable in $\fo{\setR,\setZ,+,\leq,V_b}$ can be represented by a weak RVA, the RVA symbolic representation is only used for representing sets in this logic (i.e. without $W_b$). Just remark that $\fo{\setR,\setZ,+,\leq,V_b}=\fo{\setZ,+,\leq,V_b}\uplus \fo{\setD,+,\leq}$. Finally, note that weak RVA are used in the tool \textsc{LIRA} \cite{LIRA}, whose benchmarks show very efficient computation times for sets defined in $\fo{\setR,\setZ,+,\leq}$.


\section{Towards an implementation}\label{sec:idf}


From an implementation perspective, our decomposition has been designed to fit \textsc{Genepi}'s requirements. \textsc{Genepi} \cite{genepi} is a modular framework supporting Presburger-based solvers and model-checkers, distributed under GNU Public License. Its core consists of a plugin manager, which computes generic operations (such as boolean operations, quantification, satisfiability) on sets encoded as the solutions of Presburger-like formulas. Different implementations of these operations can be used as plugins~; existing ones include \textsc{PresTAF, LIRA, LASH, MONA, OMEGA,} and \textsc{PPL}. We have begun to design a plugin for our decomposition, which uses two existing plugins : one for the integer part, and one for the decimal part.\\

Once this plugin is ready, any combination of two other plugins is possible : for example, one could try \textsc{PresTAF} over integers and \textsc{PPL} over decimals. One could even be curious and study the efficiency of two instances of \textsc{LIRA} plugins, each one working on its own part (integer or decimal). Another benefit, coming from the new decomposition of RVA, would be to use the \textsc{LASH} plugin only on one part, and manage the other one differently : this might improve the effectiveness of RVA, which are very expressive but not really efficient in practice. So far, our first tests on small conjunctions of linear constraints show execution times close to the ones of \textsc{LIRA}.\\

\medskip

What we need now for an implementation is a unique way to represent sets. Indeed, in order to avoid unduly complicated representations of sets, we have to make our representation canonical. Therefore, let us set the theoretical framework we use in practice.\\

Let $\rondZ \subseteq \partie{\setZ^n}$ and $\rondD \subseteq \partie{\setD^n}$. Notice that if $R=(Z+D_1)\cup (Z+D_2)$, then $R=Z+D$ with $D=D_1\cup D_2$~; we will always suppose that $\rondD$ is closed under union wlog. Then, notice that $R \subseteq \setR^n$ can be represented by a partially defined function $f_R$ such that :
\begin{align*}
f_R :\ & \rondZ \longrightarrow \rondD \\
& Z_i \longmapsto D_i
\end{align*}
This function's interpretation is defined as $\dsp{ \inter{f_R}=\bigcup_{i=1}^p\Bigl(Z_i+f_R(Z_i)\Bigr)}$, which matches the natural writing of $R$ introduced in section \ref{sec:composing}. Note that this representation $f_R$ is not unique.\\

For technical reasons, we extend $f_R$ to a totally defined function $\overline{f_R}$ s.t. $\overline{f_R}(Z)=\emptyset$ if $Z \notin \dom{f_R}$ and $\overline{f_R}(Z)=f_R(Z)$ otherwise. Moreover, we define the support of $\overline{f_R}$ as $\supp{\overline{f_R}} = \{Z \mid \overline{f_R}(Z) \neq \emptyset\}$. In the remainder of this paper, we will use without ambiguity the notation $f_R$ instead of $\overline{f_R}$.\\

We are now able to represent the set $R$ with a function we wish to handle. Therefore, we want to identify $f_R$ and $\inter{f_R}$ : in order to do so, this latter interpretation has to be an injection. Generally, this is not the case : using the previous definitions, we could have different writings of $\inter{f_R}$. However, if the images by $f_R$ are disjoint, then the interpretation $\inter{f_R}$ is an injection. Finally, for effectivity reasons, we will only consider functions whose support is finite. In the remainder of this section, we formalize this reasoning.\\
Let $\func_{\rondZ\rightarrow\rondD} = \{ f :\ \rondZ\longrightarrow\rondD \mid \supp{f}\mbox{ is finite}\}$.
\begin{dfn}
The \emph{interpretation function} $\inter{.}$ associates to every $f \in \func_{\rondZ\rightarrow\rondD}$ a set of real vectors defined by $\dsp{ \inter{f}=\bigcup_{Z\in \supp{f}} \Bigl(Z+f(Z)\Bigr)}$.
\end{dfn}

Notice that since $\supp{f}$ is finite, $\func_{\rondZ \rightarrow \rondD}$ do not suffice to represent every set of real vectors, as shown in the counter-example on page \pageref{ex:infini}. Let us now restrict ourselves to the functions we handle :

\begin{dfn}
An \emph{IDF (Integer-Decimal Function)} is a function $f \in \func_{\rondZ\rightarrow\rondD}$ such that $\bigcup_{Z}f(Z)=\setD^n$ and such that $Z \neq Z' \implies f(Z)\cap f(Z') = \emptyset$. We denote them all by $IDF_{\rondZ \rightarrow \rondD}= \left\lbrace f \in \func_{\rondZ\rightarrow\rondD} \mid f \mbox{ is an IDF}\right\rbrace$. We also write $\dsp{\inter{IDF_{\rondZ \rightarrow \rondD}} = \left\lbrace \inter{f} \mid f \in IDF_{\rondZ \rightarrow \rondD} \right\rbrace}$.
\end{dfn}

The sets from examples \ref{ex:set}, \ref{ex:eq}, \ref{ex:leq}, \ref{ex:sum} are represented by the following IDF :

\begin{exe}\label{ex:set2}
  The empty set $\emptyset$ is represented by the IDF $f_\bot$ defined by $f_\bot(Z)=\emptyset$ for any $Z\not=\emptyset$ and by $f_\bot(\emptyset)=\setD^n$. The set $\setR^n$ is represented by the IDF $f_\top$ (also noted $f_{\setR^n}$) defined by $f_\top(\setZ^n)=\setD^n$ and $f_\top(Z)=\emptyset$ otherwise. The set $\setZ^n$ is represented by the IDF $f_{\setZ^n}$ defined by $f_{\setZ^n}(\setZ^n)=\{\vec{0}\}$ and $f_{\setZ^n}(Z)=\emptyset$ otherwise.
\end{exe}

\begin{exe}\label{ex:eq2}
  The set $R_==\{\vec{r}\in\setR^2 \st r_1= r_2\}$ is represented by the IDF $f_=$ defined by $f_=(Z_=)=D_=$, $f_=(\emptyset)=\setD^2\moins D_=$ and $f_=(Z)=\emptyset$ otherwise, where:
  \begin{align*}
    Z_=&=\{\vec{z}\in \setZ^2 \st z_1=z_2\} &
    D_=&=\{\vec{d}\in \setD^2 \st d_1=d_2\}
  \end{align*}
\end{exe}

\begin{exe}\label{ex:leq2}
  The set $R_\leq=\{\vec{r}\in\setR^2 \st r_1\leq r_2\}$ is represented by the IDF $f_\leq$ defined by $f_\leq(Z_<)=D_>$, $f_\leq(Z_\leq)=D_\leq$ and $f_\leq(Z)=\emptyset$ otherwise where:
  \begin{align*}
    Z_<&=\{\vec{z}\in \setZ^2 \st z_1<z_2\} &
    D_>&=\{\vec{d}\in \setD^2 \st d_1>d_2\}  \\
    Z_\leq&=\{\vec{z}\in \setZ^2 \st z_1\leq z_2\} &
    D_\leq&=\{\vec{d}\in \setD^2 \st d_1\leq d_2\}
  \end{align*}
\end{exe}

\begin{exe}\label{ex:sum2}
  The set $R_+=\{\vec{r}\in\setR^3 \st r_1+r_2=r_3\}$ is represented by the IDF $f_+$ defined by $f_+(Z_0)=D_0$, $f_+(Z_1)=D_1$, $f_+(\emptyset)=\setD^3\moins (D_1\cup D_2)$ and $f_+(Z)=\emptyset$ otherwise where (intuitively $c \in\{0,1\}$ denotes a carry) :
  \begin{align*}
    Z_c &=\{\vec{z}\in \setZ^3 \st z_1+z_2+c =z_3\}\\
    D_c &=\{\vec{d}\in \setD^3 \st d_1+d_2=d_3+c\}
  \end{align*}
\end{exe}

Observe that any set in $\inter{IDF_{\rondZ_n\rightarrow\rondD_n}}$ is in $\rondZ_n\uplus\rondD_n$. The converse is obtained by proving the following proposition~:
\begin{prp}[Closure by union]\label{prp:union}
Let $R \in \inter{IDF_{\rondZ_n \longrightarrow \rondD_n}}$. Then, for any $Z \in \rondZ_n$ and $D \in \rondD_n$, we also have $R \cup (Z+D) \in \inter{IDF_{\rondZ_n \rightarrow \rondD_n}}$.
\end{prp}
\begin{proof}
We consider an IDF $f:\rondZ_n \longrightarrow \rondD_n$ such that $\inter{f}=R$ and two sets $Z \in \rondZ_n$ and $D \in \rondD_n$. We must prove that there exists an IDF $f':\rondZ_n \longrightarrow \rondD_n$ such that $\inter{f'}=R'$ with $R'=R\cup (Z+D)$. We consider the following function:
  $$\begin{array}{cccl}
    f': 
    & \rondZ_n & \longrightarrow & \rondD_n\\
    &    Z'    & \longrightarrow & \displaystyle \Bigl(f(Z')\moins D\Bigr)\bigcup_{Z''\st Z''\cup Z=Z'}\Bigl(f(Z'')\cap D\Bigr)\\
  \end{array}$$
  As expected we are going to prove that $f'$ is an IDF such that $\inter{f'}=R'$. We first show that $f'$ is an IDF. First of all observe that $\bigcup_{Z'}f'(Z')=\setD^n$. Next, let $Z_1',Z_2'\in \rondZ_n$ such that $f'(Z_1')\cap f'(Z_2')\not=\emptyset$ then either $(f(Z_1')\moins D)\cap (f(Z_2')\moins D)\not=\emptyset$ or there exists $Z_1'',Z_2''$ such that $Z_1''\cup Z=Z_1'$ and $Z_2''\cup Z=Z_2'$ and $(f(Z_1'')\cap D)\cap(f(Z_2'')\cap D)\not=\emptyset$ since the other cases are not possible. But $(f(Z_1')\moins D)\cap (f(Z_2')\moins D)\not=\emptyset$ implies $f(Z_1')\cap f(Z_2')\not=\emptyset$ and since $f$ is an IDF we get $Z_1'=Z_2'$. And $(f(Z_1'')\cap D)\cap(f(Z_2'')\cap D)\not=\emptyset$ implies $Z_1''=Z_2''$ and in particular $Z_1'=Z_2'$. We have proved that $f'$ is an IDF. Finally, equality $\inter{f'}=R'$ comes from:
  \begin{align*}
    \inter{f'}
    =&\bigcup_{Z'}(Z'+f'(Z'))\\
    =&\bigcup_{Z'} \Bigl( (Z'+(f(Z')\moins D))\\
    &\bigcup_{Z'' \st Z''\cup Z=Z'}(Z'+(f(Z'')\cap D)) \Bigr) \\
    =&\bigcup_{Z'}(Z'+(f(Z')\moins D))\bigcup_{Z''}((Z''\cup Z)+(f(Z'')\cap D))\\
    =&\bigcup_{Z''}(Z''+((f(Z'')\moins D)\cup (f(Z'')\cap D)))\\
	&\cup (Z+D\cap (\bigcup_{Z''}f(Z'')))\\
    =&\inter{f}\cup (Z+D)
  \end{align*}
\end{proof}

Hence, we have just proved the following proposition :
\begin{prp}
  $\rondZ_n\uplus\rondD_n=\inter{IDF_{\rondZ_n\rightarrow\rondD_n}}$
\end{prp}

Let us prove that this new representation is canonical~:

\begin{prp}\label{prp:unique}
For any $f_1,f_2 \in IDF_{\rondZ \rightarrow \rondD}$, $\inter{f_1}=\inter{f_2} \implies f_1=f_2$ .
\end{prp}
\begin{proof}
Consider $Z_1\subseteq \setZ^n$ and let us prove that $f_1(Z_1)\subseteq f_2(Z_1)$. Naturally, we can assume that $f_1(Z_1)\not=\emptyset$ since otherwise the inclusion is immediate. In this case, there exists $\vec{d}\in f_1(Z_1)$. As $(f_2(Z))_{Z}$ forms a sharing of $\setD^n$, there exists $Z_2$ such that $\vec{d}\in f_2(Z_2)$. Let us prove that $Z_1\subseteq Z_2$. We can assume that $Z_1\not=\emptyset$. Let $\vec{z}_1\in Z_1$ and observe that $\vec{r}_1=\vec{z}_1+\vec{d}\in \inter{f_1}$ and from $\inter{f_1}=\inter{f_2}$ we get $\vec{r}_1\in \inter{f_2}$. Thus, there exists $Z_2'$ such that $\vec{r}_1\in Z_2'+f_2(Z_2')$. Since $Z_2'\subseteq \setZ^n$ and $f_2(Z_2')\subseteq \setD^n$ we get $\vec{z}_1\in Z_2'$ and $\vec{d}\in f_2(Z_2')$. As $(f_2(Z))_{Z}$ forms a sharing of $\setD^n$ and $\vec{d}\in f_2(Z_2)\cap f_2(Z_2')$ we get $Z_2=Z_2'$. In particular $\vec{z}_1\in Z_2$ and we have proved that $Z_1\subseteq Z_2$. The other inclusion $Z_2\subseteq Z_1$ is obtained symetrically. We have proved that $Z_1=Z_2$. Therefore, $f_1(Z_1)\subseteq f_2(Z_1)$ for any $Z_1$. By symmetry we deduce that $f_1(Z)=f_2(Z)$ for any $Z$. Therefore $f_1=f_2$.
\end{proof}

Notice that in practice, this canonicity depends on how the sets in $\rondZ$ and $\rondD$ are represented. Indeed, if any of these representations are not canonical, then we can not guarantee that an $IDF_{\rondZ \rightarrow \rondD}$ will be canonical.

\section{Conclusion}
We have proposed a decomposition of three known classes into finite unions of sums of integers and decimals, providing a new characterization. This decomposition can be applied to other subsets of real vectors, and possibly yield an interest in the exploration of decidable subclasses of the full arithmetic.\\

Our main goal is to use this representation of real vectors to verify infinite systems involving counters and clocks. Indeed, we wish to extend the abilities of the tool \textsc{Fast} \cite{BFLP03} to the reals, so that it can compute exact reachability sets using acceleration techniques. A first step in such an implementation is the framework \textsc{Genepi}, allowing to solve mixed integer and real constraints defined in first-order theories. Thus, our decomposition would allow working separately on integers and reals.\\


Another advantage of our decomposition is that we can now compute operations that we did not know how to perform on certain logics. For example, there is currently no algorithm computing directly the convex hull of a set defined in $\fo{\setR,\setZ,+,\leq}$ ; but thanks to our decomposition, the problem reduces to the computation of the convex hull of Presburger-definable sets (as automata \cite{Bruyere94} or as semi-linear sets \cite{GinSpa66}), and the convex hull of sets definable in $\fo{\setD,+,\leq}$ (as finite unions of convex sets, using Fourier-Motzkin). We can push this reasoning to other symbolic representations and to other operations, such as upward or downward closure.\\

Globally, this method of separating integers and reals would speed up the software development process, because of the ease of using already existing plugins. As mentioned above, one can test the combination of any pair of plugins (provided there's at least one working on reals and another one on integers). Furthermore, a very interesting point is that a programmer can test his new plugin for real sets directly in \textsc{Genepi}, and then extend its expressivity by coupling it with \textsc{PresTAF} or another plugin handling integer sets. Obviously, the converse (extending an integer plugin to the reals) is also possible in the same fashion.\\

\bibliographystyle{latex8}
\bibliography{biblio}


\end{document}